\documentclass[11pt]{article}

\usepackage{url}
\usepackage{fullpage}
\usepackage{algcompatible,algorithm,algpseudocode,float}
\usepackage{times}
\usepackage{amsmath,amsfonts,amssymb,amsthm}
\usepackage{enumerate}
\usepackage{xcolor}
\usepackage{xspace}
\usepackage{tikz}

\usepackage[capitalise,nameinlink]{cleveref}
\crefname{lemma}{Lemma}{Lemmas}
\crefname{fact}{Fact}{Facts}
\crefname{theorem}{Theorem}{Theorems}
\crefname{corollary}{Corollary}{Corollaries}
\crefname{claim}{Claim}{Claims}
\crefname{example}{Example}{Examples}
\crefname{problem}{Problem}{Problems}
\crefname{definition}{Definition}{Definitions}
\crefname{assumption}{Assumption}{Assumptions}
\crefname{subsection}{Subsection}{Subsections}
\crefname{section}{Section}{Sections}

\usepackage{thmtools}

\usepackage{thm-restate}

\newtheorem{theorem}{Theorem}[section]
\newtheorem*{theorem*}{Theorem}

\newtheorem*{proposition*}{Proposition}

\newtheorem*{lemma*}{Lemma}

\newtheorem*{conjecture*}{Conjecture}

\newtheorem*{fact*}{Fact}

\newtheorem*{exercise*}{Exercise}

\newtheorem*{hypothesis*}{Hypothesis}

\theoremstyle{definition}

\newtheorem{exercise-easy}[theorem]{Exercise}
\newtheorem{exercise-med}[theorem]{Exercise}
\newtheorem{exercise-hard}[theorem]{Exercise$^\star$}

\newtheorem*{claim*}{Claim}

\newtheorem{remark}[theorem]{Remark}
\newtheorem*{remark*}{Remark}

\newtheorem*{observation*}{Observation}

\usepackage{prettyref}
\newcommand{\savehyperref}[2]{\texorpdfstring{\hyperref[#1]{#2}}{#2}}

\newrefformat{eq}{\savehyperref{#1}{\textup{(\ref*{#1})}}}
\newrefformat{ineq}{\savehyperref{#1}{\textup{(\ref*{#1})}}}
\newrefformat{lem}{\savehyperref{#1}{Lemma~\ref*{#1}}}
\newrefformat{def}{\savehyperref{#1}{Definition~\ref*{#1}}}
\newrefformat{thm}{\savehyperref{#1}{Theorem~\ref*{#1}}}
\newrefformat{cor}{\savehyperref{#1}{Corollary~\ref*{#1}}}
\newrefformat{cha}{\savehyperref{#1}{Chapter~\ref*{#1}}}
\newrefformat{sec}{\savehyperref{#1}{Section~\ref*{#1}}}
\newrefformat{app}{\savehyperref{#1}{Appendix~\ref*{#1}}}
\newrefformat{tab}{\savehyperref{#1}{Table~\ref*{#1}}}
\newrefformat{fig}{\savehyperref{#1}{Figure~\ref*{#1}}}
\newrefformat{hyp}{\savehyperref{#1}{Hypothesis~\ref*{#1}}}
\newrefformat{alg}{\savehyperref{#1}{Algorithm~\ref*{#1}}}
\newrefformat{rem}{\savehyperref{#1}{Remark~\ref*{#1}}}
\newrefformat{item}{\savehyperref{#1}{Item~\ref*{#1}}}
\newrefformat{step}{\savehyperref{#1}{step~\ref*{#1}}}
\newrefformat{conj}{\savehyperref{#1}{Conjecture~\ref*{#1}}}
\newrefformat{fact}{\savehyperref{#1}{Fact~\ref*{#1}}}
\newrefformat{fac}{\savehyperref{#1}{Fact~\ref*{#1}}}
\newrefformat{prop}{\savehyperref{#1}{Proposition~\ref*{#1}}}
\newrefformat{prob}{\savehyperref{#1}{Problem~\ref*{#1}}}
\newrefformat{claim}{\savehyperref{#1}{Claim~\ref*{#1}}}
\newrefformat{relax}{\savehyperref{#1}{Relaxation~\ref*{#1}}}
\newrefformat{rem}{\savehyperref{#1}{Remark~\ref*{#1}}}
\newrefformat{red}{\savehyperref{#1}{Reduction~\ref*{#1}}}
\newrefformat{part}{\savehyperref{#1}{Part~\ref*{#1}}}
\newrefformat{ex}{\savehyperref{#1}{Exercise~\ref*{#1}}}
\newrefformat{property}{\savehyperref{#1}{Property~\ref*{#1}}}
\newrefformat{case}{\savehyperref{#1}{Case~\ref*{#1}}}
\newrefformat{cat}{\savehyperref{#1}{Category~\ref*{#1}}}
\newrefformat{obs}{\savehyperref{#1}{Observation~\ref*{#1}}}
\newrefformat{cond}{\savehyperref{#1}{Condition~\ref*{#1}}}
\newrefformat{que}{\savehyperref{#1}{Question~\ref*{#1}}}

\newcommand{\eps}{\varepsilon}

\let\E\relax
\DeclareMathOperator*{\E}{\mathbb{E}}
\let\Pr\relax
\DeclareMathOperator*{\Pr}{\mathbb{P}}

\begin{document}

\author{
  Jelani Nelson\thanks{UC Berkeley. \texttt{minilek@berkeley.edu}. Supported by NSF
    award CCF-1951384, ONR grant N00014-18-1-2562, ONR DORECG award N00014-17-1-2127, an Alfred P.\ Sloan Research Fellowship, and a Google Faculty Research Award.}
    \and Huacheng Yu\thanks{Princeton University. \texttt{yuhch123@gmail.com}.}
}

\title{Optimal Bounds for Approximate Counting}

\maketitle

\thispagestyle{empty}
\setcounter{page}{0}

\begin{abstract}
Storing a counter incremented $N$ times would naively consume $O(\log N)$ bits of memory.  In 1978 Morris described the very first streaming algorithm: the ``Morris Counter'' \cite{Morris78}. His algorithm's space bound is a random variable, and it has been shown to be $O(\log\log N + \log(1/\eps) + \log(1/\delta))$ bits in expectation to provide a $(1+\eps)$-approximation with probability $1-\delta$ to the counter's value. We provide a new simple algorithm with a simple analysis showing that randomized space $O(\log\log N + \log(1/\eps) + \log\log(1/\delta))$ bits suffice for the same task, i.e.\ an exponentially improved dependence on the inverse failure probability. We then provide a new analysis showing that the original Morris Counter itself, after a minor but necessary tweak, actually also enjoys this same improved upper bound. Lastly, we prove a new lower bound for this task showing optimality of our upper bound. We thus completely resolve the asymptotic space complexity of approximate counting. Furthermore all our constants are explicit, and our lower bound and tightest upper bound differ by a multiplicative factor of at most $3+o(1)$.
 \end{abstract}

\newpage

\section{Introduction}\label{sec:intro}

Suppose one wishes to maintain an integer $N$, initialized to zero, subject to a sequence of increment operations. Maintaining this counter exactly can be accomplished using $\lceil \log_2 N\rceil$ bits. In the first example of a non-trivial streaming algorithm, Morris gave a Monte Carlo randomized ``approximate counter'', which lets one report a constant factor approximation to $N$ with large probability while using $o(\log N)$ bits of memory. His algorithm, the ``Morris Counter'', uses $O(\log\log N)$ bits \cite{Morris78}. The Morris Counter was later analyzed in more detail \cite{Flajolet85,GronemeierS09}, where it was shown that $O(\log\log N + \log(1/\eps) + \log(1/\delta))$ bits of memory is sufficient to return a $(1+\eps)$ approximation with success probability $1-\delta$; the space consumption is a random variable, and this quantity is its expectation (and in fact, the space bound holds with large probability). Further historical details can be found in \cite{Lumbroso18}.

Our main contribution is a new, simple improved algorithm and matching lower bound. In particular, we show that the correct dependence on the inverse failure probability is only doubly and not singly logarithmic. This implies for example that $O(\log\log N)$ memory suffices to have failure probability $1/\mathop{poly}(N)$, whereas previous Morris Counter analyses only guaranteed failure probability $1/\mathop{poly}(\log N)$ in such space. 

\begin{theorem}\label{thm:main}
  For any $\eps,\delta\in(0,1/2)$ there is a randomized algorithm for approximate counting which outputs $\hat N$ satisfying
  \begin{equation}
  \Pr\left(\left|N - \hat N\right| > \eps N\right) < \delta \label{eqn:approx-thm} .
\end{equation}
The memory in bits is a random variable $M$ such that for any $S > C(\log\log N + \log(1/\eps) + \log\log(1/\delta))$,\footnote{In fact our analysis is more refined and produces explicit constant factors; see Theorem~\ref{thm:space-analysis} and Remark~\ref{rem:better-space}.}
\begin{equation}
\Pr(M > S) < \exp(-C' \exp(C'' S)) .\label{eqn:space}
\end{equation}

Furthermore, our algorithm is asymptotically optimal up to a constant factor: any randomized algorithm which is promised that the final counter is in the set $\{1,\ldots,n\}$ and which satisfies \cref{eqn:approx-thm} must use $\Omega(\min\{\log n, \log\log n + \log(1/\eps) + \log\log(1/\delta)\})$ bits of memory with high probability.
\end{theorem}

Note the first term in the $\min$ of the lower bound of \cref{thm:main} is matched by a deterministic counter. We further note the space usage of the Morris Counter is also a random variable which satisfies a bound similar to \cref{eqn:space}. Next, we show that the Morris Counter itself parametrized to use the same space bound also achieves \eqref{eqn:approx-thm} as long as the counter $N$ is sufficiently large, i.e.\ at least some value $N_\delta = \Omega(\log(1/\delta))$. This is a mild restriction, since one can simply maintain a deterministic counter in parallel to the Morris Counter up to the value $N_\delta+1$. Then to answer queries, if the counter is at most $N_\delta$, we return it; else if it equals $N_\delta + 1$, we return the estimator based on the Morris Counter. As we show in the appendix, this minor tweak is necessary; without it, the Morris Counter would not achieve success probability $1-\delta$ in the desired space. We call this slight tweak ``Morris+'', which is similar to a method used in \cite{GronemeierS09}. Our next theorem provides an improved analysis of Morris+. All logarithms in this paper are base $2$, unless it is stated otherwise.

\begin{theorem}\label{thm:morris}
  For any $\eps,\delta\in(0,1/2)$ Morris+, instantiated with appropriate parameters, uses $\log\log N + 2\log(1/\eps) + \log\log(1/\delta)+O(1)$ bits of memory with high probability and outputs $\hat N$ satisfying
  \begin{equation}
  \Pr\left(\left|N - \hat N\right| > \eps N\right) < \delta \label{eqn:morris} .
\end{equation}
\end{theorem}

Though we provide two proofs of the same upper bound, we believe both have value. One perhaps pedagogical advantage of Theorem~\ref{thm:main} is that the new algorithm we provide is designed with the analysis in mind, leading to an overall proof of our novel optimal upper bound that is both short and intuitive. That is, one reads the argument and feels they clearly understand ``why'' the upper bound is what it is. Meanwhile, the advantage of Theorem~\ref{thm:morris} is that it provides a tight analysis of an algorithm commonly used in practice, albeit at the pedagogical cost that the proof of the theorem boils down to a technical calculation, and the reason the final bound comes out the way it does is arguably less intuitive.

Given that most modern machines have much more than $\log N$ bits of memory for even for $N$ being on the order of the number of particles in the universe, one might wonder whether approximate counting is of real importance or merely a purely intellectual pursuit. An application to keep in mind is not that there is merely one counter, but we may wish to maintain many such counters. In a real such application the number of approximate counters could be very large, and so cutting the number of bits per counter by even a constant factor could be of value. Indeed this was Morris' own original motivation: he needed to keep track of not only one counter, but $26^3$ counters, to keep trigram counts as part of the spellchecker \textsf{typo} \cite{Lumbroso18}. An example of a real such scenario in the modern day is the implementation of the ``Least Frequently Used'' (LFU) cache eviction policy in Redis, one of the most popular in-memory databases. The Redis implementation of this eviction policy needs to keep track of a counter for each key in the database, corresponding to the number of times it has been queried recently. To save memory, these counters are in fact implemented as approximate counters \cite{Redis}.

This motivating perspective also reveals that typically the memory requirement to calculate the state transition of the approximate counter after an increment, or to answer a query, is much less important; rather, minimizing the memory required to maintain program state is of higher practical relevance, as that affects total storage. Furthermore, if we are maintaining $M$ counters then it is natural to want $\delta \ll 1/M$ so that each counter is approximately correct with high probability. If $M$ is large, then requiring $\log(1/\delta) \ge \log M$ bits per counter may provide no benefit over a naive $\log N$ bit counter for realistic values of $N$.

In addition to potential practical relevance, from a theoretical perspective ``maintaining a counter'' is a natural problem and as such the Morris Counter has found applications to other streaming problems. For example, Jayaram and Woodruff showed that for $p\in(0,1]$ an approximate counter can be used effectively as a subroutine in an algorithm for approximating the $p$th moment of an insertion-only stream up to $1+\eps$ in $\tilde O(1/\eps^2 + \log n)$ bits of space \cite{JayaramW19}, improving over a derandomization of an algorithm of Indyk that uses $O(\eps^{-2}\log n)$ bits \cite{Indyk06,KaneNW10}. Approximate counting also finds use in approximating large frequency moments \cite{AlonMS99,GronemeierS09}, approximate reservoir sampling \cite{GronemeierS09}, approximating the number of inversion when streaming over a permutation \cite{AjtaiJKS02}, and $\ell_1$ heavy hitters in insertion-only streams \cite{BhattacharyyaDW19}.

\subsection{Comparison with previous bounds from \cite{Flajolet85}}
As we discuss in \cref{sec:overview}, the Morris Counter works by storing a counter $X$ and incrementing it with probability $1/(1+a)^X$ per update for some parameter $a>0$. The work \cite{Flajolet85} characterized the behavior of the Morris algorithm {\it exactly} when $a=1$. Unfortunately, the Morris Counter for $a=1$, which uses $O(\log\log N)$ bits of memory with high probability (which is $O(\log\log N + \log\log(1/\delta))$ for $\delta = 1/\mathop{poly}(N)$), does not enjoy constant factor approximation with success probability any better than a constant even for large $N$, let alone with probability $1 - 1/\mathop{poly}(N)$. This failure of the Morris Counter to achieve very high success probability for $a=1$ is implied by the exact characterization of the algorithm given in \cite{Flajolet85} itself; Proposition 3 of that work implies that the probability that $X$ fails to be in the interval $[\log_2 N - C, \log_2 N + C]$ equals a constant (depending on $C$), and $X$ being in that interval is required for the Morris Counter to provide a $2^C$-approximation. Thus, the failure probability when $a=1$ is not even $o(1)$. Our \cref{thm:morris} reveals though that the Morris Counter with $a = \Theta(1/\log N)$ {\it does} achieve failure probability $1/\mathop{poly}(N)$, which is ``for free'' (up to a constant factor) compared with $a=1$ since this smaller setting of $a$ still only requires the Morris Counter to use $O(\log\log N)$ bits of memory with high probability.

\cite{Flajolet85} does have some discussion on using smaller $a$. Specifically, \cite[Section 5]{Flajolet85} mentions that if one wants error better than the case $a=1$ to estimate $N$, one can either average independent counters or change base, and that the former has ``an effect similar to'' the latter. A variance bound is then given for estimating $N$ when using arbitrary $a$. This variance bound seems to reveal though that the effects of averaging versus changing base are not similar from a computational complexity perspective: the former requires averaging $\Omega(1/\eps^2)$ copies of the counter, blowing up the space complexity by $1/\eps^2$. The latter leads to a space bound depending only on $O(\log(1/\eps))$. Both yield $O(\log(1/\delta))$ space dependence on the failure probability $\delta$. Equation (46) of \cite{Flajolet85} does give an explicit sum-product formula for the exact probabilities $P_{n,\ell}$ that the counter exactly equals $\ell$ after $n$ increments, but this formula is not readily prescriptive for how $a$ should be set in order to achieve relative error $1+\eps$ with failure probability $\delta$.

\subsection{Overview of approach}\label{sec:overview}
We first explain the idea behind the Morris Counter. The traditional, deterministic and exact counter stores an integer $X$, initialized to zero. After every increment to $N$, we increment $X$ with probability $1.0^X$, i.e.\ we always increment it. Thus we can ``estimate'' $N$ as $X$, and this estimator has zero variance and is unbiased, at the cost of using $\log N$ memory. Morris instead increments $X$ with probability $0.5^X$; this trades off variance for memory. Specifically, one can show that $\E[2^X - 1] = N$, though the variance only satisfies $\mathrm{Var}[2^X - 1] = N(N-1)/2$. A natural idea of Morris is then to change the base of the exponential when deciding the probability to increment $X$, which turns out to provide a smooth tradeoff between memory and space consumption. Specifically, for any $a>0$ if incrementing $X$ with probability $1/(1+a)^X$, the expression $a^{-1}((1+a)^X - 1)$ is an unbiased estimator of $N$ with variance $aN(N-1)/2$ (we call the Morris Counter with this parameterization ``Morris($a$)''). Setting $a=2\eps^2\delta$, one obtains the guarantee \cref{eqn:approx-thm} via Chebyshev's inequality. Note that the space consumption $S:=\lceil \log_2 X\rceil$ is a random variable, but is at most $O(\log\log N + \log(1/\eps) + \log(1/\delta))$ with high probability. This is because for $C > 2$, once $X > Z := (\log N/(2\eps^2\delta))^C$, by a union bound the probability that any of the remaining at most $N$ increments causes $X$ to increment even once more is at most $N (1+2\eps^2\delta)^{-Z} < e^{-(\log N/(2\eps^2\delta))^{C-1}} < N^{-\omega(1)}$ (using that $(1-r)^{1/r} < 1/e$ for $r>0$). Thus, with high probability the Morris Counter uses at most $O(\log Z) = O(\log\log N + \log(1/\eps) + \log(1/\delta))$ bits of memory.

We now describe our new algorithm. First, we consider a promise decision problem: given some $T>1$ and $\eps\in(0,1)$, decide whether $N < (1-\eps/10)T$ or $N>(1+\eps/10)T$ when promised that one of the two holds. We can solve this decision problem as follows. We store a counter $Y$ in memory, initialized to $0$. Set $\alpha = \min\{1, C \log(1/\eta)/(\eps^2 T)\}$ for $C$ a large constant and $\eta\in(0,1)$ a parameter to be set. For each increment to $N$, if $Y \le \alpha T$ then increment $Y$ with probability $\alpha$; else do nothing. At query time, we declare $N > (1+\eps/10)T$ iff $Y > \alpha T$. A Chernoff bound shows that this procedure is correct with probability at least $1-\eta$. Furthermore the memory consumed is guaranteed to be $O(\log(\alpha T)) = O(\log(1/\eps) + \log\log(1/\eta))$.

Now to solve the full approximate counting problem, and not just the decision problem, we solve multiple instantiations of the above promise problem in sequence, where in iteration $j$ we use the threshold $T_j = (1+\eps)^j$ and increment probability $\alpha_j = \min\{1, C\log(1/\eta_j)/(\eps^3 T_j)\}$ for $\eta_j < C\delta/j^2$ (chosen so that by a union bound, the probability that we ever fail to solve the promise problem in any iteration $j$ is at most $\sum_j \eta_j \le \delta$). When $Y$ reaches the value $\alpha_j T_j$, we increase $j$ and correspondingly set $Y\leftarrow \lfloor Y\cdot \alpha_{j+1} /\alpha_j \rfloor$ (which is ``correct in expectation'', since the number of increments we would have done in expectation with parameter $\alpha_{j+1}$ is an $\alpha_j/\alpha_{j+1} \approx 1+\eps$ factor less). To answer a query for $N$, we simply return $T_j$. The adjustment from $\eps^2$ to $\eps^3$ in $\alpha_j$ is for technical reasons (see the proof of \cref{thm:correctness}).

We next provide an improved analysis of Morris' original algorithm. To do so, we define the random variable $Z_i$ to be the number of increments that Morris($a$), run for an infinite number of increments, would have its counter $X$ equal to $i$ before incrementing to $X=i+1$. Then $Z_i$ is a geometric random variable with parameter $1/(1+a)^i$, and we are able to show the desired behavior of Morris($a$) by proving concentration bounds on prefix sums of the $Z_i$ via analyzing its moment-generating function.

Our new lower bound comes from showing that a randomized approximate counter using space $S$ can be made deterministic with no increased space cost at the cost of increasing its failure probability by factors that grow with $S$. If $S$ is smaller than a certain threshold (the lower bound we are trying to prove), this argument leads to a correct space-$o(\log n)$ deterministic algorithm for the problem, which is impossible, and thus the space-$S$ algorithm for $S$ so small could not have existed.

\subsection{Notation}\label{sec:notation}
We use $C, C', C''$ to denote universal positive constants, which may change from line to line. We also use $A \pm B$ to denote a value in the interval $[A-B, A+B]$, with $D = A\pm B$ signifying $D\in [A-B, A+B]$. As mentioned, we also use ``Morris($a$)'' to refer to the Morris Counter parameterized to increment $X$ with probability $1/(1+a)^X$.

\section{Improved upper bound for approximate counting}\label{sec:alg}
In \cref{sec:new-algo} we describe and analyze our new algorithm for approximate counting with space complexity $O(\log\log N + \log(1/\eps) + \log\log(1/\delta))$. We then show that this upper bound is achieved by the original Morris Counter itself in \cref{sec:morris}.

\subsection{New algorithm description and analysis}\label{sec:new-algo}

\begin{algorithm}[!h] 
  \caption{Approximate counting algorithm.} \label{algo:main}
  \begin{algorithmic}[1]
    \Procedure{$\mathsf{ApproxCount}$}{$\eps,\delta$}
    \State \textbf{\underline{Init():}}
    \State $\eta \leftarrow \delta, X_0 \leftarrow \lceil\ln_{1+\eps}(C\ln(1/\eta)/\eps^3)\rceil$
    \State $Y\leftarrow 0, X\leftarrow X_0 , \alpha\leftarrow 1, T\leftarrow \lceil (1+\eps)^X\rceil$
    \Statex
    \State \textbf{\underline{Increment():}}
    \State with probability $\alpha$, update $Y\leftarrow Y + 1$
    \If {$Y > \alpha T$}
      \State $X\leftarrow X+1$
      \State $T \leftarrow \lceil (1+\eps)^X\rceil, \eta \leftarrow \frac{\delta}{X^2}$
      \State $\alpha_{\text{new}} \leftarrow \frac{C\ln(1/\eta)}{\eps^3 T}$
      \State $Y\leftarrow \lfloor Y\cdot \alpha_{\text{new}}/\alpha \rfloor$
      \State $\alpha\leftarrow \alpha_{\text{new}}$
    \EndIf      
    \Statex
    \State \textbf{\underline{Query():}}
    \If {$X=X_0$}
      \State \Return $Y$
    \Else
      \State \Return $T$
    \EndIf
    \EndProcedure
  \end{algorithmic}
\end{algorithm}

We describe our full approximate counting algorithm in \cref{algo:main}. The counter is initialized via the \textbf{Init()} procedure, and each increment to $N$ and query for an estimate of $N$ are described in the pseudocode, following the ideas set forth in \cref{sec:overview}. \cref{thm:correctness} shows that the relative error of the output of \cref{algo:main} is $1+O(\eps)$ with probability $1-O(\delta)$. \cref{eqn:approx-thm} follows by adjusting $\eps,\delta$ by a constant factor. Our variable $X$ is quite similar to that of the Morris Counter: it represents (an approximation to) $\log_{1+\eps} N$. The main difference is that whereas the Morris Counter decides to increment $X$ based on flipping a number of coins depending on $X$ itself, we use an auxiliary counter $Y$ to guide when $X$ should be incremented. 

First we define some notation that will be useful for the proof. We divide the algorithm's execution into epochs $k=0,1,2,\ldots,$ corresponding to the value of $X-X_0$. We mark the end of an epoch immediately before line 8 is about to execute, and the beginning of the new epoch immediately after line 13 has completed executing. During a given epoch, we let $T_k, \alpha_k, \eta_k$ be the corresponding values of $T,\alpha, \eta$ set in lines 7--12 of \cref{algo:main}. For example, $T_0 = 1, \alpha_0 = 1, \eta_0 = \delta$. We also define $Y_k$ to be the value of $Y$ when epoch $k$ begins, so that $Y_0 = 0$ and $Y_k$ for $k>0$ is set in line 11 of \cref{algo:main}. To be precise, a particular epoch is said to begin after \cref{algo:main} completes lines 4 or 12, and it ends at line 6 when the if statement triggers. We say that $N$ becomes a certain value once \textbf{Increment()} has been called that number of times, {\it and} the most recent call completed.

\begin{theorem}\label{thm:correctness}
  There is a universal constant $C'>0$ such that $\forall \eps,\delta\in(0,1/2)$, the output $\hat N$ of \textbf{Query()} in \cref{algo:main} satisfies  $\Pr(|\hat N - N| > C'\eps N) < C'\delta$.
\end{theorem}
\begin{proof}

  We first note that while remaining in epoch $0$, i.e.\ as long as $1\le N\le T_0$, $Y$ stores $N$ exactly and thus our output is exactly correct. Our focus is thus on the case of larger $N$.

  For $k\ge 0$, define the event $\mathcal E_k$ that once we enter epoch $k$, the number of increments to $N$ before we advance to the next epoch is $T_k - T_{k-1} \pm \eps^2T_{k-1}$ (where we use the convention $T_{-1} = 0$). We henceforth condition on the event $\wedge_{k\geq 0} \mathcal E_k$. 
  Since the $T_r$ are in geometric series with base $1+\eps$ (up to $\pm 1$ due to rounding), we have $\sum_{r=0}^k (T_r - T_{r-1} \pm \eps^2 T_{r-1}) \subseteq (1\pm 1.5\eps)T_k$, i.e., only after $(1\pm1.5\eps)T_k$ increments to $N$, could the algorithm possibly be in epoch $k$.
  Thus, if $k^*$ is the final epoch when \textbf{Query()} is called, we have $\hat N=T_{k^*}$ and $N=(1\pm1.5\eps)T_{k^*}$.
  That is, $\hat N=\frac{1}{1\pm 1.5\eps} N$, which implies $|\hat N-N|\leq C\eps N$ when $\eps<1/2$.

  We finally bound
  $$
  \Pr\left(\bigwedge_{k=0}^{\infty} \mathcal E_k\right) = 1 - \Pr\left(\bigvee_{k=0}^{\infty} \neg\mathcal E_k\right) \ge 1 - \sum_{k=0}^{\infty} \Pr(\neg\mathcal E_k) .
  $$
  $\Pr(\neg \mathcal E_0) = 0$, so we focus on $k\ge 1$. Note $Y_k = \lfloor (\lfloor \alpha_{k-1}  T_{k-1}\rfloor +1) \cdot (\alpha_k/ \alpha_{k-1})\rfloor$, which is $\alpha_k T_{k-1}\pm O(1)$. The new threshold for $Y$ to enter epoch $k+1$ is $\lfloor \alpha_k T_k\rfloor + 1$, which thus requires $\alpha_k (T_k-T_{k-1})\pm O(1)$ more increments to $Y$, which is
  \begin{equation}
    \eps \alpha_k T_{k-1}\pm O(1),\label{eqn:num-increments}
  \end{equation}
  since $T_k-T_{k-1} = \eps T_{k-1}\pm O(1)$ and $\alpha\leq 1$. 
  To upper bound the probability that we already advance to the next epoch after calling \textbf{Increment()} $t_1 := T_k-T_{k-1}-\eps^2 T_{k-1}$ times, it suffices to consider the following question: If we increment $Y$ with probability $\alpha_k$ independently for each of the $t_1$ \textbf{Increment()} calls, what is the probability that we increment $Y$ \emph{at least} $\varepsilon\alpha_kT_{k-1}-O(1)$ times.\footnote{Note that in the actual execution of the algorithm, not all $t_1$ calls increment $Y$ with probability $\alpha_k$, e.g., if we have advanced to the next epoch already, then the probability becomes $\alpha_{k+1}$. Nevertheless, the probability that we advance to the next epoch after $t_1$ \textbf{Increment()} calls is the same if we increment $Y$ with $\alpha_k$ probability for each call, since it does not matter if we have \emph{already} advanced to the next epoch.}

  The expected number of times $Y$ is incremented is $$\alpha_k t_1=\eps\alpha_k T_{k-1} - \eps^2\alpha_k T_{k-1} \pm O(1),$$ which is $\Theta(\ln(1/\eta_k)/\eps^2)$. Advancing to the next epoch thus implies deviating from the expectation by more than $\eps^2\alpha_k T_{k-1}\pm O(1)$, i.e., $\eps$ times the expectation. The Chernoff bound implies that the probability of this occurring is at most $\eta_k$. A similar calculation shows that the probability that we have {\it not} advanced to the next epoch after calling \textbf{Increment()} $t_2 := T_{k} - T_{k-1} + \eps^2 T_{k-1}$ times. Thus $\Pr(\neg \mathcal E_k) \le 2\eta_k$. Thus $\Pr(\vee_{k\geq 0} \neg\mathcal E_k) \le 2\sum_k \eta_k = 2\sum_k \delta/(k+1)^2 = O(\delta)$.
\end{proof}

\begin{remark}\label{rem:implementation}
  Before we give the space analysis, the astute reader may notice that $T$ itself is ideally approximately $N$ and thus should require $\Theta(\log N)$ bits to store. A similar statement could be made about the Morris Counter: the output is ultimately given as $a^{-1}((1+a)^X - 1)$ (see \cref{sec:overview}), which is also $\Theta(\log N)$ bits. The key is that in implementation, we never actually store $T$: we only store $X$. Then our answer to a query is only to return $X$, which will be an additive $O(1)$ approximation to $\log_{1+\eps} N$ with high probability, which is enough for the querying party to specify an approximation to $N$. Similarly, $\delta$ is never stored or even given to the algorithm, but rather the input should be $\Delta$ such that $\delta = 2^{-\Delta}$, and only $\Delta$ is ever stored. Also, the correctness analysis only requires that $\alpha$ be {\it at least} the value in line 10 and not exactly that (to apply the Chernoff bound effectively). Thus $\alpha$ can be rounded up to the nearest inverse power of $2$ so that $\alpha = 2^{-t}$ and only $t$ need be stored consuming only $\log t = \log\log(1/\alpha)$ bits. We can then generate a $\mathsf{Bernoulli}(\alpha)$ random variable (line 6) by flipping a fair coin $t$ times and returning $1$ iff all flips were heads; this takes $1$ bit to keep track of the \textsf{AND} and $\log t$ bits to keep track of the number of flips made so far. $\eta$ also need not be stored explicitly since its value is implicit from other stored values (namely $X$, $\eps$, and $\Delta$).

  Of course the situation is even simpler in models of computation other than word RAM, such as a finite automaton or branching program:  then program constants need not be stored in memory (they only affect the transitions), and only the variables $X, Y$ contain program state that needs to be stored. Furthermore, what is most important from the perspective of the practical motivation in \cref{sec:intro} when running a system storing many approximate counters is the number of bits required to maintain program state; it is reasonable to assume in practical applications that $O(\log N)$ bit registers are available to be used temporarily while processing updates and queries, which could lead to faster and simpler implementation.
  \end{remark}

\begin{theorem}\label{thm:space-analysis}
    For any $\eps,\delta\in(0,1/2)$, the probability that \cref{algo:main} needs more than
    \[
        \log\log N+\log\log(1/\delta)+3\log(1/\eps)+\Omega(t)
    \]
    bits of memory after $N$ increments is at most $(\eps/N)^{2^{t}}$, for any $t\geq C\cdot(\log\log\log N+\log\log(1/\eps))$, where $C$ is a sufficiently large constant.
\end{theorem}
To see that this theorem implies the space bound stated in Theorem~\ref{thm:main}, for any $S>C(\log\log N+\log (1/\varepsilon)+\log\log(1/\delta))$ for a sufficiently large $C$, we have $t>(C-3)(\log\log N+\log (1/\varepsilon)+\log\log(1/\delta))>S/2$.
Hence, the probability that we use more than $S$ bits of memory after $N$ increments is at most
\[
    (\varepsilon/N)^{2^t}\leq 2^{-2^t}\leq \exp(-C'\exp(C''S)),
\]
for some constants $C',C''>0$.

\begin{proof}
  As described in \cref{rem:implementation}, \cref{algo:main} only explicitly stores two variables $X$ and $Y$.
  When $X=X_0$, $Y$ is between $0$ and $T=O(\log(1/\delta)/\eps^3)$.
  In this case, storing $Y$ takes $$\log\log(1/\delta)+3\log(1/\eps)+O(1)$$ bits.
  When $X=X_0+k$ for $k\geq 1$ (i.e., in epoch $k$), $Y$ is between $\alpha_kT_{k-1}-O(1)$ and $\alpha_kT_k+O(1)$.
  In this case, storing $Y$ takes
  \begin{align*}
    \log(\alpha_k(T_k-T_{k-1})+O(1))&\leq \log\log(1/\eta)+2\log(1/\eps)+O(1) \\
    &\leq \log\log(1/\delta)+2\log(1/\eps)+2\log\log X+O(1)
  \end{align*}
  bits.
  Thus, provided that $X\leq X_{\max}$, Algorithm 1 uses at most
  \begin{equation}\label{eqn_space_given_xmax}
    \max\{\log X_{\max}, \log(1/\eps)\}+\log\log(1/\delta)+2\log(1/\eps)+2\log\log X_{\max}+O(1)
  \end{equation}
  bits.
  In the following, we show that the final $X$ is small with high probability.

  We will show that once we reach an epoch $k$ for $k$ large (corresponding to $X=X_0+k$), with high probability we will never advance to epoch $k+1$. Indeed, the probability that we do advance is the probability that $Y$ increments at least $\eps \alpha_k T_{k-1}\pm O(1)$ times over the at most $N$ remaining calls to \textbf{Increment()} (see \cref{eqn:num-increments}). By a union bound over all $(\eps\alpha_kT_{k-1}+O(1))$-subsets of the remaining increments, the probability that this occurs is at most
  \begin{align*}
    \binom N{\eps \alpha_k T_{k-1}\pm O(1)} \cdot \alpha_k^{\eps \alpha_k T_{k-1}\pm O(1)} &\le \left(\frac{2e N}{\eps T_{k-1}}\right)^{\eps \alpha_k T_{k-1}\pm O(1)}\\
    {}&\le \left(\frac{C'N}{\eps (1+\eps)^{X}}\right)^{\Theta(\log(X^2/\delta)/\eps^2)}.
  \end{align*}

  For $X\geq 2\log_{1+\eps} (N/\eps)$, it is at most
  \begin{align*}
    \left(\frac{C'N}{ \eps(1+\eps)^{X}}\right)^{\Theta(\log(X^2/\delta)/\eps^2)}  &\le \left(\frac 1{(1+\eps)^X}\right)^{\Omega(1/\eps^2)} \\
    {}&\le (e^{-\Theta(\eps X)})^{\Omega(1/\eps^2)}\\
    {}&\le e^{-\Omega(X)}.
    \end{align*}
    
    By setting $X_{\max}=\Theta(2^{t}\log_{1+\eps}(N/\eps))$ for some integer $t\geq C\cdot(\log\log\log N+\log\log(1/\eps))$, i.e., $t\geq \log\log X_{\max}+\log\log(1/\eps)$,
    \begin{align*}
        \log X_{\max}&\leq\log\log N+\log(1/\eps)+\log\log(1/\eps)+t+O(1) \\
        &= \log\log N+\log(1/\eps)+\Theta(t).
    \end{align*}
    By Equation~\eqref{eqn_space_given_xmax}, the probability that \cref{algo:main} needs more than
    \[
        \log\log N+\log\log(1/\delta)+3\log(1/\eps)+\Omega(t)
    \]
    bits of space is at most
    \[
        \left(\frac{\eps}{N}\right)^{2^{t}}.
    \]
  \end{proof}
  
  \begin{remark}
    In the proof, we assumed that the algorithm allocates exactly $\log X_{\max}$ bits to store $X$, and then we bounded the probability that $X$ exceeds $X_{\max}$ after $N$ increments.
    This assumption requires us to have an upper bound on $N$ in advance.
    In general, when an upper bound on $N$ is unknown, we will have to store variable $X$ that is also unbounded, and dynamically allocate bits to the counter.
    This can be done by first encoding $\lceil \log X\rceil$ using $O(\log\log X)$ bits, then encoding $X$ using $\lceil \log X\rceil$ bits.
    Our proof gives the same space bound in this case.
  \end{remark}

  \begin{remark}\label{rem:better-space}
The source of the constant factor ``3'' multiplying $\log(1/\eps)$ in the space complexity is due to the cubic dependence of $\alpha_{\text{new}}$ on $1/\eps$ in \cref{algo:main}. This cubic dependence was due to the proof structure of \cref{thm:correctness}: we conditioned on the events $\mathcal E_k$ that we spent a concentrated amount of time in each epoch. To show that this happens with high probability, we performed a union bound over all epochs. We feel this structure makes the proof more intuitive, though it comes at the cost of a worsened constant factor. One can show that the algorithm is still in fact correct with $\alpha_{\text{new}}$ depending only quadratically on $1/\eps$ by proving concentration only on the total time spent on all the epochs combined, as opposed to union bounding over epochs separately, by using an argument similar to what we will see shortly in \cref{sec:morris}. One can also see empirically via implementation that the algorithm of this section and Morris+ behave nearly identically, including the constant factor (see \cref{sec:implementation}).
  \end{remark}

\begin{remark}
  Our approximate counter is fully mergeable \cite{AgarwalCHPWY13}.
  That is, given two counters $(X_1,Y_1)$ and $(X_2,Y_2)$, which approximate two (unknown) numbers $N_1$ and $N_2$ respectively, they can be merged into a single data structure $(X, Y)$ that follows the same distribution as if it was incremented exactly $N_1+N_2$ times so that nothing is lost in the parameters $\eps$ and $\delta$ (the Morris Counter enjoys this same benefit \cite[Section 2.1]{CormodeY20}).
  To see this, observe that each epoch of our algorithm uses sampling, and the sampling rate is non-increasing.
  Assuming $X_1\leq X_2$, we can simulate $N_1$ extra increments to the second counter by another subsampling with the correct probabilities.
  More specifically, the first counter is in epoch $k_1=X_1-X_0$, and we know the sampling probabilities $\alpha_0,\ldots,\alpha_{k_1}$, and the exact number of increments that survived the sampling (caused $Y_1$ to increment) in each epoch.
  We are going to insert all the survivors to the second counter, which currently have sampling probability $\alpha_{k_2}$ for $k_2=X_2-X_0$.
  For each survivor in epoch $i$ (for $0\leq i\leq k_1$), we increment $Y_2$ with probability $\alpha_{k_2}/\alpha_i$.
  Then effectively, we increment $Y_2$ with probability $\alpha_{k_2}$ for each of the \emph{original} $N_1$ increments.
  Whenever $Y_2$ reaches the threshold $\alpha T$, we increment $X_2$, update $Y_2$, and adjust the probabilities.
  Hence, the final $(X_2, Y_2)$ has the same distribution as if it was incremented a total of $N_1+N_2$ times.
\end{remark}

\subsection{Morris Counter improved analysis}\label{sec:morris}
Here we analyze the Morris($a$) algorithm for some $a\in (0, 1)$, in which $X$ is incremented with probability $(1+a)^{-X}$ and we output $\hat N = ((1+a)^X-1)/a$. 
When the total number of increments $N$ is at most $8/a$, the value of the counter can be explicitly maintained in addition to the Morris Counter, which costs at most $\log (1/a)+O(1)$ bits of space.
In the following, we assume $N$ is at least $8/a$; this is not a serious limitation since we can maintain a separate counter exactly, deterministically up until this value (the ``Morris+'' modification described in \cref{sec:intro}).

Let us consider Morris($a$) on an infinite sequence of increments.
For any $i\geq 0$, $X$ exceeds $i$ with probability $1$.
Let $Z_i\geq 1$ be the random variable denoting the number of increments it takes for $X$ to increase from $i$ to $i+1$.
Since when $X=i$, each increment causes $X$ to increase with probability $p_i=(1+a)^{-i}$, $Z_i$ follows the geometric distribution
\[
  \Pr[Z_i=l]=(1-p_i)^{l-1}p_i.
\]
Therefore, we have
\[
  \E[Z_i]=1/p_i=(1+a)^i,
\]
and
\[
  \E[e^{tZ_i}]=\sum_{l\geq 1}e^{tl}(1-p_i)^{l-1}p_i=\frac{e^tp_i}{1-e^t(1-p_i)},
\]
for any $t$ such that $e^t(1-p_i)<1$.

Next, let $\eps<1/2$, we bound 
\begin{equation}\label{eqn_upper}
  \Pr\left[\sum_{i=0}^kZ_i\geq (1+\eps)\sum_{i=0}^k 1/p_i\right].
\end{equation}
Following the proof of Chernoff bound, for $t$ such that $e^t(1-p_k)<1$, we have
\begin{align*}
  \E\left[e^{t\sum_{i=0}^kZ_i}\right]&=\prod_{i=0}^k\E\left[e^{tZ_i}\right]\\
  &=\frac{e^{(k+1)t}\prod_{i=0}^kp_i}{\prod_{i=0}^k(1-e^t(1-p_i))} \\
  &=\frac{e^{(k+1)t}(1+a)^{-k(k+1)/2}}{\prod_{i=0}^k(1-e^t(1-p_i))}.
\end{align*}
By Markov's inequality,
\begin{align*}
  \eqref{eqn_upper}&\leq \frac{\E\left[e^{t\sum_{i=0}^kZ_i}\right]}{e^{t(1+\eps)\sum_{i=0}^k1/p_i}} \\
  &=\frac{\E\left[e^{t\sum_{i=0}^kZ_i}\right]}{e^{t(1+\eps)((1+a)^{k+1}-1)/a} } \\
  &=\frac{e^{(k+1)t}(1+a)^{-k(k+1)/2}}{e^{t(1+\eps)((1+a)^{k+1}-1)/a} \prod_{i=0}^k(1-e^t(1-(1+a)^{-i}))}.
\end{align*}
Now set $t=\ln\left(\frac{1}{1-\frac12\eps(1+a)^{-k}}\right)$, which satisfies $e^t(1-p_k)<1$, we have
\begin{align*}
  \eqref{eqn_upper}&\leq (1+a)^{-k(k+1)/2}\cdot \left(1-\frac12\eps(1+a)^{-k}\right)^{-(k+1)+(1+\eps)((1+a)^{k+1}-1)/a}\cdot\prod_{i=0}^k\frac{1}{1-\frac{1-(1+a)^{-i}}{1-\frac12\eps(1+a)^{-k}}} \\
  &= (1+a)^{-k(k+1)/2}\cdot \left(1-\frac12\eps(1+a)^{-k}\right)^{-(k+1)+(1+\eps)((1+a)^{k+1}-1)/a}\cdot\prod_{i=0}^k\frac{1-\frac12\eps(1+a)^{-k}}{(1+a)^{-i}-\frac12\eps(1+a)^{-k}} \\
  &= (1+a)^{-k(k+1)/2}\cdot \left(1-\frac12\eps(1+a)^{-k}\right)^{(1+\eps)((1+a)^{k+1}-1)/a}\cdot\prod_{i=0}^k\frac{1}{(1+a)^{-i}(1-\frac12\eps(1+a)^{-k+i})} \\
  &\leq e^{-\frac12\eps(1+a)^{-k}(1+\eps)((1+a)^{k+1}-1)/a}\cdot\prod_{i=0}^k\frac{1}{1-\frac12\eps(1+a)^{-k+i}}.
\end{align*}
By the fact that $1/(1-z)\leq e^{z+z^2}$ for all $0<z<1/2$, 
\begin{align*}
    \eqref{eqn_upper}&\leq e^{-\frac12\eps(1+a)^{-k}(1+\eps)((1+a)^{k+1}-1)/a}\cdot e^{\sum_{i=0}^k(\frac12\eps(1+a)^{-k+i}+\frac14\eps^2(1+a)^{-2k+2i})} \\
    &= e^{-\frac12\eps(1+a)^{-k}\left((1+\eps)((1+a)^{k+1}-1)/a-\sum_{i=0}^k((1+a)^{i}+\frac12\eps(1+a)^{-k+2i})\right)} \\
    &\leq e^{-\frac12\eps(1+a)^{-k}\left((1+\eps)((1+a)^{k+1}-1)/a-(1+\frac12\eps)((1+a)^{k+1}-1)/a)\right)} \\
    &=e^{-\frac14\eps^2(1+a)^{-k}((1+a)^{k+1}-1)/a}.
\end{align*}
For $k>\frac{1}{a}$, we have $\eqref{eqn_upper}\leq e^{-\eps^2/8a}$.

Similarly, we next bound
\begin{equation}\label{eqn_lower}
  \Pr\left[\sum_{i=0}^kZ_i\leq (1-\eps)\sum_{i=0}^k 1/p_i\right].
\end{equation}
By Markov's inequality,
\begin{align*}
  \eqref{eqn_lower}&=\Pr\left[e^{-t\sum_{i=0}^kZ_i}\geq e^{-t(1-\eps)\sum_{i=0}^k 1/p_i}\right] \\
  &\leq \frac{\E\left[e^{-t\sum_{i=0}^kZ_i}\right]}{e^{-t(1-\eps)\sum_{i=0}^k 1/p_i}} \\
  &=\frac{e^{-t(k+1)}(1+a)^{-k(k+1)/2}}{e^{-t(1-\eps)((1+a)^{k+1}-1)/a}\prod_{i=0}^k(1-e^{-t}(1-p_i))}.
\end{align*}
Now set $t=\ln(1+\frac12 \eps(1+a)^{-k})$, we have
\begin{align*}
  \eqref{eqn_lower}&\leq (1+a)^{-k(k+1)/2}\cdot \left(1+\frac12 \eps(1+a)^{-k}\right)^{-(k+1)+(1-\eps)((1+a)^{k+1}-1)/a}\cdot \frac{1}{\prod_{i=0}^k\left(1-\frac{1-(1+a)^{-i}}{1+\frac12 \eps(1+a)^{-k}}\right)} \\
  &=(1+a)^{-k(k+1)/2}\cdot \left(1+\frac12 \eps(1+a)^{-k}\right)^{-(k+1)+(1-\eps)((1+a)^{k+1}-1)/a}\cdot \frac{1}{\prod_{i=0}^k\left(\frac{(1+a)^{-i}+\frac12 \eps(1+a)^{-k}}{1+\frac12 \eps(1+a)^{-k}}\right)} \\
  &=(1+a)^{-k(k+1)/2}\cdot \left(1+\frac12 \eps(1+a)^{-k}\right)^{(1-\eps)((1+a)^{k+1}-1)/a}\cdot \frac{1}{\prod_{i=0}^k(1+a)^{-i}\left(1+\frac12 \eps(1+a)^{-k+i}\right)} \\
  &=\left(1+\frac12 \eps(1+a)^{-k}\right)^{(1-\eps)((1+a)^{k+1}-1)/a}\cdot \frac{1}{\prod_{i=0}^k\left(1+\frac12 \eps(1+a)^{-k+i}\right)}.
\end{align*}
By the fact that $1/(1+z)\leq e^{-z+z^2}$ for $z\geq 0$, we have
\begin{align*}
  \eqref{eqn_lower}&\leq e^{\frac12 \eps(1+a)^{-k}(1-\eps)((1+a)^{k+1}-1)/a}\cdot e^{\prod_{i=0}^k\left(-\frac12 \eps(1+a)^{-k+i}+\frac14 \eps^2(1+a)^{-2k+2i}\right)}\\
  &=e^{\frac12 \eps(1+a)^{-k}\left((1-\eps)((1+a)^{k+1}-1)/a+\prod_{i=0}^k\left(-(1+a)^{i}+\frac12 \eps(1+a)^{-k+2i}\right)\right)} \\
  &\leq e^{-\frac14 \eps^2(1+a)^{-k}((1+a)^{k+1}-1)/a}.
\end{align*}
When $k>\frac1a$, this is at most $e^{-\eps^2/8a}$.

Therefore, for any $k>1/a$, with probability at least $1-e^{-\eps^2/8a}$, we have
\[
  \left|\sum_{i=0}^kZ_i-((1+a)^{k+1}-1)/a\right|\leq \eps((1+a)^{k+1}-1)/a.
\]
Now fix any $N>8/a$, let $k_1$ be the largest $k$ such that $(1+\eps)((1+a)^{k+1}-1)/a<N$, $k_2$ be the smallest $k$ such that $(1-\eps)((1+a)^{k+1}-1)/a\geq N$.
We have $k_1,k_2>1/a$, then we apply the above inequality to $k_1$ and $k_2$, and by union bound, with probability at least $1-2e^{-\eps^2/8a}$, we have both
\[
  \sum_{i=0}^{k_1}Z_i\leq (1+\eps)((1+a)^{k_1+1}-1)/a<N,
\]
i.e., $X>k_1$ after $N$ increments, and
\[
  \sum_{i=0}^{k_2}Z_i\geq (1-\eps)((1+a)^{k_2+1}-1)/a\geq N,
\]
i.e., $X\leq k_2$ after $N$ increments.
Therefore, $((1+a)^X-1)/a$ is a $(1\pm 2\eps)$ approximation of $N$ with probability $1-2e^{-\eps^2/8a}$.

By setting $a=\eps^2/(8\ln(1/\delta))$, the space usage of Morris($a$) is $\log\log N+\log (1/a)+O(1)=\log\log N+2\log(1/\eps)+\log\log(1/\delta)+O(1)$ bits with high probability, and outputs a $(1\pm 2\eps)$ approximation with probability $1-2/\delta$.
By reparametrizing, we prove Theorem~\ref{thm:morris}.

\begin{remark}
While it may be possible to improve the constant factor ``8'' in the exponent of the tail bound above, note that this constant in turn only affects the setting of $a$ by a constant factor, and the space complexity of Morris($a$) only depends logarithmically on $1/a$. Thus, any improvement to the factor $8$ can only improve the analysis of the space complexity by an additive constant.
  \end{remark}

\begin{remark}\label{rem:eric}
After seeing our proof, Eric Price pointed out that it can be made even more succinct as follows: one can show that geometric random variables are ``subgamma'', so that a sum of geometric random variables (as in Eqs.~\eqref{eqn_upper} and \eqref{eqn_lower}) is subgamma with appropriate parameters (see \cite[Section 2.4]{BoucheronLM13} for the definition and relevant properties of subgamma random variables).
\end{remark}

\section{Space lower bound}\label{sec_lb}
\newcommand{\cC}{\mathcal{C}}
\newcommand{\dett}{\mathrm{det}}
Here we prove the matching lower bound for approximate counters. Our lower bound states that even if the algorithm's memory usage is a random variable which only has a {\it small} chance of being small (i.e.\ we allow it to use arbitrarily large memory with  large probability $1- \sqrt \delta$), it still cannot satisfy \cref{eqn:approx-thm}.
\begin{theorem}
Fix $\eps,\delta\in(0,1/2)$ and integer $n$.
Let $\cC$ be an approximate counter which outputs $\hat N$ satisfying 
\[
	\Pr(|N-\hat N|>\eps N)<\delta,
\]
for all $N\in \{1,\ldots, n\}$, and uses no more than $S$ bits of space with probability at least $\sqrt{\delta}$.
We must have 
\begin{align*}
    S&\geq \min\{\log n-O(1),\max\{\log\log n+\log(1/\eps)-O(\log\log(1/\eps)), \log\log (1/\delta)-O(\log\log\log(1/\delta))\},
\end{align*}
which is at least $\Omega(\min\{\log n,\log\log n+\log(1/\eps)+\log\log (1/\delta)\})$.
\end{theorem}
The first observation is that conditioned on using no more than $S$ bits of space, we have
\[
	\Pr(|N-\hat N|>\eps N\mid \textrm{use at most $S$ bits})<\delta/\Pr(\textrm{use at most $S$ bits})\leq \sqrt{\delta}.
\]
Hence, we may assume that $\cC$ always uses at most $S$ bits of space, at the cost of increasing the failure probability to $\sqrt{\delta}$, which is inconsequential since the dependence on $\delta$ in the space bound is $\log\log(1/\delta)$.
In the following, we assume that $\cC$ never uses more than $S$ bits.

Let $T=\lfloor \min\{n/4, \frac{\log (1/\delta)}{4\log\log(1/\delta)}\}\rfloor$.
Then for every $N=1,\ldots,T/2$, $\mathcal{C}$ outputs $\hat N$ that is less than $T$ with probability $1-\delta$, and for every $N=2T,2T+1,\ldots,4T$, $\mathcal{C}$ outputs $\hat N$ that is at least $T$ with probability $1-\delta$.
In particular, $\cC$ distinguishes $N\in [1, T/2]$ and $N\in [2T, 4T]$ with probability $1-\delta$.
In the following, we show that any $\cC$ that distinguishes the two cases with probability $1-\delta$ must use $\log T-O(1)$ bits of space.
We assume for contradiction that $S\leq \log (T/4)$.

First, let us consider the following ``derandomization'' of $\cC$.
$\cC$ uses no more than $S$ bits of space, hence, it has at most $2^S$ different memory states.
When \textbf{Init()} is called, the algorithm generates a (possibly random) initial memory state.
Each time \textbf{Increment()} is called, the algorithm examines the current state and updates the memory to a possibly different state (and possibly randomly).
Let the ``deterministic'' version of the algorithm $\cC_\dett$ have the same query algorithm as $\cC$, but when \textbf{Init()} or \textbf{Increment()} is called, it examines the current state and the distribution of the new state (or the initial state) according to $\cC$; instead of updating the memory according to this distribution, $\cC_\dett$ always updates it to the state with the \emph{highest} probability in this distribution (in case of tie, pick the lexicographically smallest).

Now let us analyze the error probability of $\cC_{\dett}$.
The initialization and increment algorithms are called exactly $N+1$ times in total.
Since $\cC_{\dett}$ picks the state with the highest probability each time, which has probability at least $2^{-S}$, the probability that the execution of $\cC$ follows the exact same path as $\cC_{\dett}$ is at least
\[
	\left(2^{-S}\right)^{N+1}.
\]
Therefore, conditioned on the execution of $\cC$ following the same path, its error probability is at most
\[
	\delta\cdot \left(2^{S}\right)^{N+1}.
\]
When $N\leq 4T$, it is at most
\[
	\delta\cdot (T/4)^{4T+1}\leq \delta\cdot (\log(1/\delta)/(16\log\log(1/\delta)))^{\log(1/\delta)/\log\log(1/\delta)+1}<1/3 .
\]
That is, the error probability of $\cC_{\dett}$ is at most $1/3$, for every $N\in[1,T/2]\cup[2T,4T]$.

On the other hand, since both initialization and increment algorithms are deterministic, we may apply an argument similar to the ``pumping lemma'' for DFAs.
Since $2^S\leq T/4$, there exists $1\leq N_1<N_2\leq T/2$ such that $\cC_{\dett}$ reaches the same memory state after $N_1$ or $N_2$ increments.
Again by the fact that the increment algorithm is deterministic, $\cC_{\dett}$ must reach the same memory state after $N_1+k(N_2-N_1)$ increments, for all integer $k\geq 0$.
In particular, there exists $N_3\in [2T, 4T]$ such that $\cC_{\dett}$ reaches this memory state after $N_3$ increments.
However, by the assumption of the algorithm, the query algorithm distinguishes between $N_1$ increments and $N_3$ increments with probability at least $2/3$, which is impossible as the algorithm reaches the same memory state in the two cases.
This proves that $S\geq \log T-O(1)$, i.e., 
\begin{equation}\label{eqn_lb_delta}
	S\geq \min\{\log n-O(1), \log\log (1/\delta)-O(\log\log\log(1/\delta)))\}.
\end{equation}

Finally, we show that $S\geq \min\{\log n,\log\log n+\log (1/\eps)\}-O(1)$ as long as $\delta\in(0,\sqrt{1/2})$.
Let $N_j=\left\lceil (e^{16\eps j}-1)/\eps\right\rceil$, and consider incrementing the counter $N_j$ times for an unknown $j$.
Observe that for $j\geq 0$, we have
\begin{align*}
	(1-\eps)N_{j+1}-(1+\eps)N_j&\geq(1-\eps)(e^{16\eps (j+1)}-1)/\eps-(1+\eps)(e^{16\eps j}-1)/\eps-(1+\eps) \\
	&=((1-\eps)e^{16\eps}-(1+\eps)) e^{16\eps j}/\eps-(3+\eps) \\
	&\geq ((1-\eps)(1+16\eps)-(1+\eps))/\eps-(3+\eps) \\
	&=11-17\eps \\
	&> 0.
\end{align*}
Therefore, for every $j\geq 0$ and $j\leq (1/16\eps)\ln (\eps n+1)$ (hence, $N_j\leq n$), $\cC$ recovers $j$ with probability $1-\delta>1/5$, if the counter is incremented $N_j$ times.
By fixing the random bits used by $\cC$, at least $1/5$ fraction of such $j$ is successfully recovered.
The algorithm must reach a different final state for all such $j$, implying that
\[
	2^S\geq \frac{1}{5}\cdot (1/16\eps)\ln (\eps n+1)=\Omega((1/\eps)\log (\eps n+1)).
\]
When $\eps<1/n$, it is $\Omega((1/\eps)(\eps n))=\Omega(n)$, and
\[
	S\geq \log n-O(1).
\]
When $1/n\leq \eps<1/\sqrt{n}$, we have
\[
	S\geq \log (1/\eps)-O(1)\geq \log(1/\eps)+\log\log n-O(\log\log(1/\eps)).
\]
When $\eps\geq 1/\sqrt{n}$, we have
\[
    S\geq \log (1/\eps)+\log\log(\eps n)-O(1)\geq \log(1/\eps)+\log\log n-O(1).
\]
In all three cases, the bounds imply 
\begin{equation}\label{eqn_lb_eps}
	S\geq \min\{\log n-O(1),\log\log n+\log (1/\eps)-O(\log\log(1/\eps))\}.
\end{equation}

Finally, by \eqref{eqn_lb_delta} and \eqref{eqn_lb_eps}, we conclude that
\begin{align*}
	S&\geq \min\{\log n-O(1),\max\{\log\log n+\log(1/\eps)-O(\log\log(1/\eps)), \log\log (1/\delta)-O(\log\log\log(1/\delta))\} \\
	&=\Omega(\min\{\log n,\log\log n+\log(1/\eps)+\log\log (1/\delta)\}).
\end{align*}
proving the claimed lower bound.

\section{Philosophical digression: the value of implementation}\label{sec:implementation}

\begin{figure}[H]
  \begin{center}
    \scalebox{.25}{\includegraphics{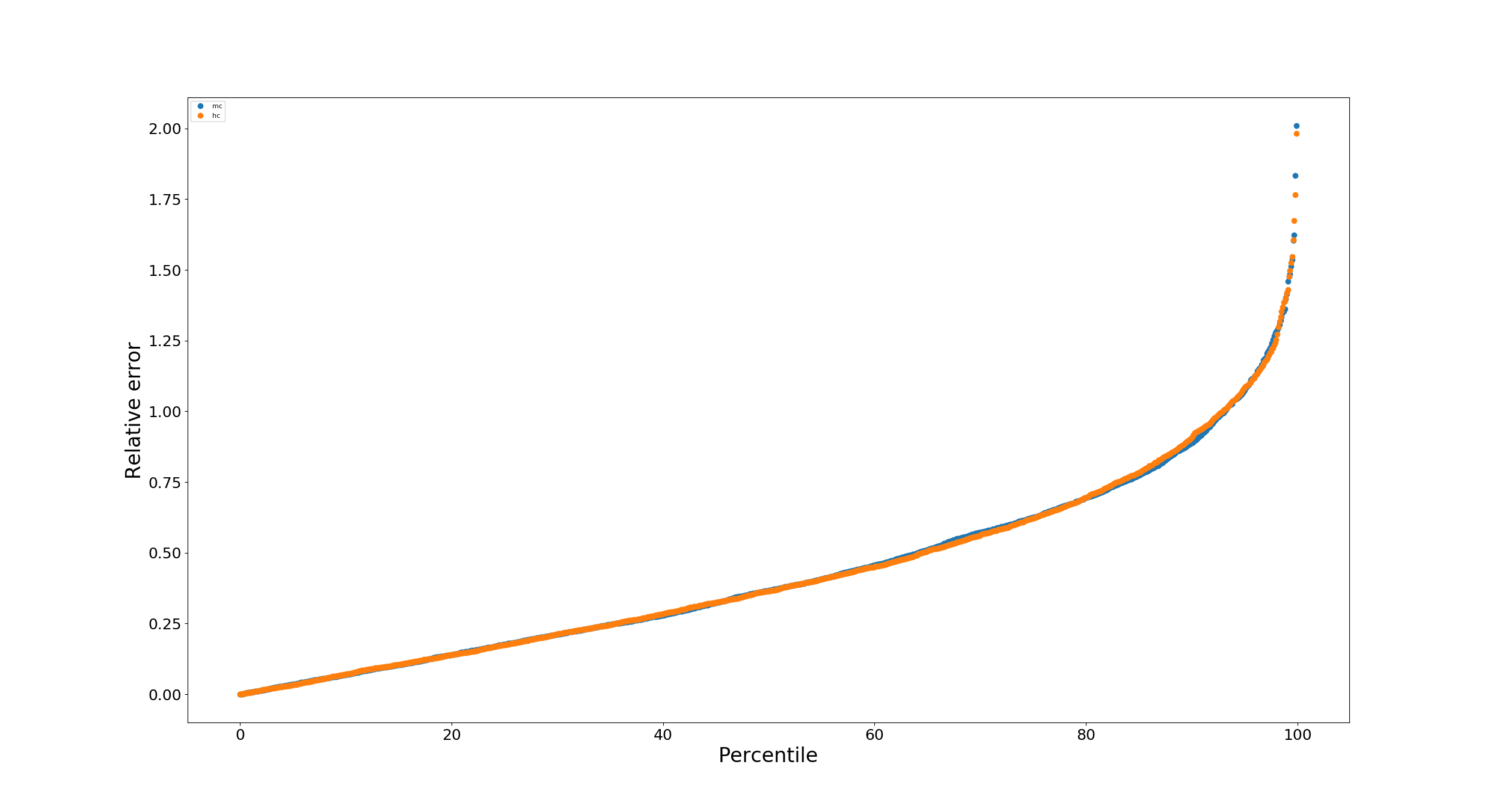}}
  \caption{Results of experimental comparison of the Morris counter and a simplified version of the algorithm of \cref{sec:new-algo}.}\label{fig:experiments}
  \end{center}
\end{figure}

We share in this section a historical note on the development of this work, which may serve the reader as evidence of the value of implementation. Chronologically, we first developed and analyzed the algorithm of \cref{sec:new-algo} and proved the lower bound in \cref{sec_lb}. In the days afterward, excited by the prospect of having a new and improved algorithm for such a fundamental problem, we implemented the Morris Counter as well as (a simplified version of) the algorithm of \cref{sec:new-algo} (and this simplified algorithm is itself similar to the algorithm of \cite{Csuros10}) to compare. We ran several experiments. In one, we did the following 5,000 times for each algorithm, parameterized to use only 17 bits of memory: pick a uniformly random integer $N\in[500000,999999]$ (thus a 20-bit number) and perform $N$ increments. The results of this experiment are in \cref{fig:experiments}. The orange plot represents our algorithm, and the blue plot is the Morris Counter. For each respective algorithm's color, a dot plotted at point $(x,y)$ means that in $x\%$ of the trial runs (out of 5,000), the relative multiplicative error of the algorithm's estimate was $y\%$ or less. In other words, we plotted the empirical CDFs of the relative errors of each algorithm. For example, the plot indicates that neither algorithm ever had relative error more than $2.37\%$ in 5,000 runs. The experimental results are plainly apparent: the two algorithms' empirical performances are {\it nearly identical}! Witnessing this plot convinced us that the previously known analyses of the Morris Counter, an algorithm that has been known for over 40 years and taught in numerous courses, were most likely suboptimal and that the Morris Counter itself is most likely an optimal algorithm for the problem. With the confidence gained from the experimental results, we sought a new and improved analysis of the Morris Counter and succeeded. Thus it seems from this anecdote, implementation can sometimes be valuable even for purely theoretical work.

  \section*{Acknowledgments}
  We thank Eric Price for pointing out the content of \cref{rem:eric} and allowing us to include it here.

  \newcommand{\etalchar}[1]{$^{#1}$}

\appendix

\section{Tweaking the Morris Counter is necessary}
In this section we show that the modification from the vanilla Morris Counter to ``Morris+'' described in \cref{sec:intro} is necessary. Recall the modification: when using Morris($a$), we maintain a deterministic counter $X'$ in parallel. During increments, we process the increment both by Morris($a$) and by deterministically incrementing $X'$, unless its value is $N_a+1$ in which case we do not alter it. During queries, if $X' \le N_a$, we return $X'$; otherwise we return the estimator from Morris($a$) based on $X$. We set $N_a = 8/a$, as suggested by the analysis in \cref{sec:morris}.

We now show that if one {\it does not} modify the Morris Counter but simply uses Morris($a$) for $a = \eps^2/(8\ln(1/\delta))$ as suggested in \cref{sec:morris}, then when $\delta<\eps^{8/3}c^2/16$, $\eps<1/4$ and the counter value equals $N = N'_a := c\eps^{4/3} /a\geq 2$ for a constant $c\le 2^{-8}$, the probability that the Morris Counter outputs an estimator $\hat N<(1-\eps)N$ is much larger than $\delta$. Note that our analysis requires switching from a deterministic counter to the Morris Counter when $N \ge \Omega(1/a)$ and not $\Omega(\eps^{4/3}/a)$, but the impact on memory complexity is at most a factor of three (and less as $N$ grows): using a deterministic counter up until $N=r$ requires an additional $\lceil\log_2 r\rceil$ bits. Thus the difference between $r = c_1/a$ versus $r = c_2\eps^{4/3}/a$ is the difference between $\log r = \log(c_1) + 3 + \log\log(1/\delta) + 2\log(1/\eps)$ versus $\log r = \log(c_2) + \log\log(1/\delta) + \frac{2}{3}\log(1/\eps)$; i.e.\ the dependence on $\log(1/\eps)$ differs by a factor of three. Thus our analysis here shows that for small $\delta$, our choice of transition point $r = 8/a$ from a deterministic counter to using the Morris Counter is almost optimal, up to affecting the memory by a multiplicative factor of at most three.

We now show why Morris($a$) will fail with probability much larger than $\delta$. Consider the event $\mathcal E$ that the Morris Counter increments $X$ in the first $t$ increment operations, and its value remains equal to $t$ in the last $N-t$ increments, for $t= \lfloor\ln(1+(1-2\eps)\eps^{4/3} c)/\ln(1+a)\rfloor$.
Recall the estimator is $\hat N = a^{-1}((1+a)^X - 1)$. Thus conditioned on $\mathcal E$,
\begin{align*}
  \hat N&= \frac1a \cdot\left((1+a)^t-1\right) \\
        &\le \frac 1a \cdot\left(1+(1-2\eps)\eps^{4/3} c - 1\right)\\
        &= (1-2\eps)N\\
        &<(1-\eps)N
\end{align*}

On the other hand, note that $t\geq\ln(1+(1-2\eps)\eps^{4/3} c)/\ln(1+a)-1\geq \frac{1}{a}\ln(1+(1-2\eps)\eps^{4/3}c)-1$, and $t\leq N$.
The probability of $\mathcal E$ is at least
\begin{align*}
	\Pr[\mathcal E] &= \prod_{i=0}^{t-1}(1+a)^{-i}\cdot \left(1-(1+a)^{-t}\right)^{N-t} \\
	&\geq (1+a)^{-t^2}\cdot \left(1-(1+a)\left(1+(1-2\eps)\eps^{4/3} c\right)^{-1}\right)^{N-\frac{1}{a}\ln(1+(1-2\eps)\eps^{4/3} c)+1} \\
	&= (1+a)^{-t^2}\cdot \left(\frac{1+(1-2\eps)\eps^{4/3} c-(1+a)}{1+(1-2\eps)\eps^{4/3}c}\right)^{\frac{1}{a}\left(c\eps^{4/3}-\ln(1+(1-2\eps)\eps^{4/3} c)\right)+1} \\
	&\geq (1+a)^{-N^2}\cdot \left(\frac{\eps^{4/3}c}{4}\right)^{\frac{1}{a}\left(c\eps^{4/3}-\ln(1+(1-2\eps)\eps^{4/3} c)\right)+1},
	\intertext{which by the fact that $\ln(1+x)\leq x$ and $\ln(1+x)\geq x-x^2/2$ for $x<1$, is}
	&\geq \frac{\eps^{4/3}c}{4}\cdot e^{-aN^2}\cdot \left(\frac{\eps^{4/3}c}{4}\right)^{\frac{1}{a}\left(c\eps^{4/3}-(1-2\eps)\eps^{4/3} c+((1-2\eps)\eps^{4/3} c)^2/2\right)} \\
	&=\frac{\eps^{4/3}c}{4}\cdot e^{-\frac{1}{a}(\eps^{4/3}c)^2}\cdot e^{-\frac{\ln(4/(\eps^{4/3}c))}{a}\left(2\eps^{7/3}c+\eps^{8/3}c^2/2\right)} \\
	&\geq \frac{\eps^{4/3}c}{4}\cdot e^{-\frac{\eps^2}{32a}}\cdot e^{-\frac{\eps^2}{a}\left(4\eps^{1/3}c\ln(4/(\eps^{4/3}c)\right)} \\
	&\geq \frac{\eps^{4/3}c}{4}\cdot e^{-\frac{\eps^2}{16a}} \\
	&=\frac{\eps^{4/3}c}{4}\cdot \sqrt{\delta}.
\end{align*}

When ${\delta}< \eps^{8/3}c^2/16$, this is larger than $\delta$.
Therefore, Morris$(a)$ fails to provide a $(1-\eps)$-approximation for $N$ with probability at least $\delta$.

\end{document}